%% file: article.tex
\documentclass[a4paper, 10pt]{article}
\usepackage[margin=1in]{geometry}
\usepackage{srcltx,graphicx,epstopdf}
\usepackage{amsmath, amssymb, amsthm}
\usepackage{color}
\usepackage{multirow}
\usepackage{overpic}
\usepackage{bm}
\usepackage{subfig} 
\usepackage{hyperref}
\usepackage{booktabs}
\usepackage{tikz}

\usepackage{cleveref}
\crefname{section}{Section}{Sections}
\crefname{subsection}{Subsection}{Subsections}
\crefname{appendix}{Appendix}{Appendix}
\crefname{figure}{Figure}{Figures}
\crefname{table}{Table}{Tables}
\crefname{property}{Property}{Properties}
\crefname{theorem}{Theorem}{Theorem}
\crefname{criterion}{criterion}{criteria}

\graphicspath{{images/}}

\newtheorem{theorem}{Theorem}
\newtheorem{lemma}[theorem]{Lemma}

\newtheorem{property}[theorem]{Property}

\newtheorem{remark}{Remark}

\newcommand\bA{{\bf A}}

\newcommand\bD{{\bf D}}
\newcommand\bDn{{\tilde{\bD}}}
\newcommand\bKn{{\tilde{\bK}}}

\newcommand\bK{{\bf K}}
\newcommand\bM{{\bf M}}

\newcommand\bT{{\bf T}}

\newcommand\bbR{\mathbb{R}}
\newcommand\bbN{\mathbb{N}}

\newcommand\bbG{\mathbb{G}}
\newcommand\bbS{\mathbb{S}}

\newcommand\br{\boldsymbol{r}}
\newcommand\bv{\boldsymbol{v}}
\newcommand\bw{\boldsymbol{w}}
\newcommand\bx{\boldsymbol{x}}

\newcommand\bxi{\bm{\xi}}

\newcommand\bmeta{\bm{\eta}}
\newcommand\dd{\,\mathrm{d}}

\newcommand\constantsymbol{{\frac{m}{\hat{h}^3}}}
\newcommand\Temperature{RT}
\newcommand\zz{{\mathfrak{z}}}
\newcommand\Li{\mathrm{Li}}
\newcommand\Lihalf[1]{{\Li_{\frac{#1}{2}}}}

\newcommand\PN{{$P_N$ }}
\newcommand\MN{{$M_N$ }}

\newcommand\MP[1]{{${M\!P}_{#1}$ }}
\newcommand\MPN{{\MP{\!N}}}
\newcommand\Mone{{$M_1$ }}

\newcommand\HMP[1]{{${H\!M\!P}_{#1}$ }}
\newcommand\HMPN{{\HMP{\!N}}}

\newcommand\He{{\mathrm{He}}^{[\bm\eta]}}
\newcommand\mHe{{\mathcal{H}}^{[\bm\eta]}}

\newcommand\mE{{\mathcal{E}}}

\newcommand\mN{{\mathcal{N}}}
\newcommand\mEn{{\tilde{\mE}}}

\newcommand\mB{{\mathcal{B}}}
\newcommand\mS{{\mathcal{S}}}

\newcommand\mK{{\mathcal{K}}}
\newcommand\mI{{\mathcal{I}}}

\newcommand\weight{\omega^{[\bmeta]}}
\newcommand\weightf{\tilde{\omega}^{[\bmeta]}}

\newcommand\pl{\phi^{[\bmeta]}}
\newcommand\psl{\psi^{[\bmeta]}}

\newcommand\pd[2]{\dfrac{\partial {#1}}{\partial {#2}}}
\newcommand\od[2]{\dfrac{\dd {#1}}{\dd {#2}}}

\numberwithin{equation}{section}

\newcommand\delete[1]{}

\title{Direct Flux Gradient Approximation to Close Moment Model for
  Kinetic Equations}

\author{ Ruo Li\thanks{CAPT, LMAM \& School of Mathematical Sciences,
    Peking University, Beijing, China, email: {\tt
      rli@math.pku.edu.cn}},~~ Weiming Li\thanks{Laboratory of
    Computational Physics, Institute of Applied Physics and
    Computational Mathematics, Beijing, China, email:
    \tt{liweiming@pku.edu.cn}},~~ Lingchao Zheng\thanks{School of
    Mathematical Sciences, Peking University, Beijing, China, email:
    {\tt lczheng@pku.edu.cn}} } \date{\today}

\begin{document}
\maketitle

\input{abstract}
\input{intro}
\input{model1d}
\input{example1d}
\input{modelmultid}
\input{examplemultid}
\input{conclusion}

\bibliographystyle{abbrv}
\bibliography{article}
\end{document}

%% file: abstract.tex
\begin{abstract}
  To close the moment model deduced from kinetic equations, the
  canonical approach is to provide an approximation to the flux
  function not able to be depicted by the moments in the reduced
  model. In this paper, we propose a brand new closure approach with
  remarkable advantages than the canonical approach. Instead of
  approximating the flux function, the new approach close the moment
  model by approximating the flux gradient. Precisely, we approximate
  the space derivative of the distribution function by an ansatz which
  is a weighted polynomial, and the derivative of the closing flux is
  computed by taking the moments of the ansatz. Consequently, the
  method provides us an improved framework to derive globally
  hyperbolic moment models, which preserve all those conservative
  variables in the low order moments. It is shown that the linearized
  system at the weight function, which is often the local equilibrium,
  of the moment model deduced by our new approach is automatically
  coincided with the system deduced from the classical perturbation
  theory, which can not be satisfied by previous hyperbolic
  regularization framework.  Taking the Boltzmann equation as example,
  the linearlization of the moment model gives the correct
  Navier-Stokes-Fourier law same as that the Chapman-Enskog expansion
  gives. Most existing globally hyperbolic moment models are
  re-produced by our new approach, and several new models are proposed
  based on this framework.

  \vspace*{4mm}
  
\end{abstract}


%% file: intro.tex
\section{Introduction}

The kinetic equation is the evolution equation of particle
distribution function in phase space, and describes the motion of
particles and their interactions with each other or with the background
medium. Common examples of the kinetic equation
include the Boltzmann equation for rarefied gas, the radiative
transfer equation for photon transport, the vlasov equation for
plasmas, etc. They have wide applications in fields like rarefied gases,
microflow, semi-conductor device simulation, radiation astronomy,
optical imaging, and so on. Among the numerous methods developed to solve the
kinetic equations
\cite{broadwell1964study,Bird,hayakawa2007coupled,larsen2010advances}, 
moment methods are attractive due to their
clear physical interpretation and high efficiency in the transitional
regimes \cite{Grad,minerbo1978maximum,levermore1996moment,
McDonald,Fan_new,Koellermeier,MPN}.

The moment method approximates the
original kinetic equation by studying the evolution of a finite number
of moments of the distribution function. In gas kinetic theory, it was
first proposed in Grad's seminal paper \cite{Grad} in 1949, in which the
notable Grad's 13 moment system is also presented. For any moment
model, the evolution equations of the lower order moments rely on the
higher order moments, hence a moment closure is needed, which is the
central problem of the moment method.  Different types of moment
closures have been developed, resulting in many different kinds of
moment models. These models could roughly be divided into two types:
the first type are called hyperbolic or first order PDEs, these
include Grad's 13 moment system \cite{Grad}, the maximum entropy moment
system \cite{levermore1996moment}, 
approximate maximum entropy moment systems \cite{McDonald}, the
\PN and \MN models for radiative transfer
\cite{pomraning1973equations,minerbo1978maximum}, etc. The second kind
are parabolic type or second order PDEs, these include regularized
moment methods of various kinds \cite{Struchtrup2003, Struchtrup2004,
Torrilhon2004} and the diffusion approximation
for the radiative transfer equation. This paper focus on discussing
the first order moment models for their several advantages, as pointed
out in \cite{mc2013affordable}. First, as
only the first derivatives are computed, an extra order of spatial
accuracy can potentially be obtained for a given stencil. Secondly,
numerical solution of first order systems are less sensitive to grid
irregularities \cite{mcdonald2009use}, which often occur for practical
situations where there are complex geometries, or when adaptive mesh
refinement are used. For this first order system of equation, an
important property is its hyperbolicity, which is the necessary
condition for the local well-posedness of problem with Cauchy data.
However, obtaining global hyperbolicity for moment systems is not
trivial. In \cite{Grad13toR13} the authors pointed out that Grad's 13
moment system is not hyperbolic even in any neighbourhood of the
Maxwellian. Over the past decade, much effort has been devoted to
obtaining globally hyperbolic moment models. Levermore proved in
\cite{levermore1996moment} that the maximum entropy moment models are
globally hyperbolic. The approximate affordable robust version of the
14 moment system of the hierarchy studied by Levermore, which is
proposed by McDonald and Torrilhon in \cite{McDonald}, is almost
globally hyperbolic. There has also been on-going effort on hyperbolic
regularization of Grad-type moment models, such as the hyperbolic
regularization proposed by Cai et. al \cite{Fan,Fan_new}, and the
quadrature-based regularization method proposed in
\cite{Koellermeier}. Based on understanding of these regularization
methods, \cite{framework} and  \cite{Fan2015} proposed a general
framework to deduce globally hyperbolic moment systems, where the
distribution function is approximated by weighted polynomials. The
framework has been applied to various fields to obtain globally
hyperbolic moment models in a routine way. 

However, for some moment systems, applying the framework proposed in
\cite{framework, Fan2015} is a naive way. For instance, in
\cite{HMPN}, a direct application of the framework changes the \Mone
model, which is already globally hyperbolic and based on which the
\MPN model is derived.  On the other hand, in
\cite{di2016quantum,HMPN}, it was pointed out that the moment model
derived by directly applying the framework in \cite{framework,
Fan2015} lead to incorrect NSF law and Eddington approximation, when a
one-step Maxwellian iteration is applied. There has been some recent
progress in this direction. In \cite{HMPN}, for the \MPN model of the
radiative transfer equation \cite{MPN}, a new hyperbolic
regularization was proposed, in which the distribution function and
its derivative were approximated in different spaces.  The \HMPN
model, based on this hyperbolic regularization, no longer had the
above disadvantages. This inspires us to propose a new hyperbolic
regularization for the moment models of kinetic equation.

In this paper, we propose a new approach to close hyperbolic moment
models derived for kinetic equations. We directly approximate the flux
gradient to give the moment closure. Precisely, we approximate the
derivative of the distribution function by a weighted polynomial. In
traditional moment methods, the moment closure is often given by
approximating the distribution function, and some of them suffer from
lack of hyperbolicity \cite{Grad13toR13, QuantumGrad13, MPN}. The new
closure approach provide us an improved framework \cite{framework,
  Fan2015} to carry out model reduction for generic kinetic
equations. We first prove that using the improved framework, the
moment system is globally hyperbolic, as long as the weight function
satisfies the basic requirement of positivity. Moreover, for an $N$-th
order moment system, moments with orders from $0$ to $N-1$ are
conservative, without considering the right hand side. Additionally,
the moment model by the improved framework does not change the
consequence of Maxwellian iteration, which means the NSF laws are
preserved. As mentioned above as comparison, direct application of
hyperbolic regularization might not preserve this property
\cite{di2016quantum, HMPN}. Furthermore, we show that the HME model
\cite{Fan,Fan_new} for the Boltzmann equation, the \PN model, the \MN
model, and the \HMPN model for the radiative transfer equation
\cite{pomraning1973equations, minerbo1978maximum, HMPN, 3DMPN} can be
regarded as special cases derived by the new approach. Using the new
closure approach, we propose new globally hyperbolic moment systems
for kinetic equations, such as \HMPN model for monochromatic case and
a moment model for the Boltzmann-Peierls equation. The improved
framework is extended to 3-D case naturally, where the advantages
discussed above are all preserved. Using the improved framework, we
derive a new hyperbolic moment model for the Grad's 13-moment system
for quantum gas.

The remaining parts of the paper is organized as follows. In
\cref{sec:model1d}, we introduce our new closure approach in the 1-D
case as the essential ingredients of the improved framework, and some
properties are proved. In \cref{sec:example1d}, we list some classical
moment models and put them into the category of the improved
framework, while some new models are proposed. In
\cref{sec:modelmultid}, the framework is extended to 3-D case, and the
3-D examples are shown in \cref{sec:examplemultid}. The paper ends
with a conclusion in \cref{sec:conclusion}.


%% file: model1d.tex

\section{Model in one-dimensional case} \label{sec:model1d}
In this section, we first introduce the moment method for kinetic equation 
and our new framework in 1-D case, 
where the kinetic equation is often given by 
\begin{equation}
  \label{eq:kinetic1d}
  \pd{f}{t} + v(\xi)\pd{f}{x} = \mathcal{S}(f),\quad
  t\in\bbR^+,x\in\bbR,v(\xi)\in\bbR,\xi\in\bbG\subset\bbR.
\end{equation}
$f(t,x,\xi)$ is the distribution function depending on time $t$,
spatial variable $x$ and velocity-related variable $\xi$.
The $k$-th moment is defined as $E_{k}(t,x):=\langle \tau(\xi)v^kf\rangle$,
where $\tau(\xi)>0$, and  
\[
\langle \cdot \rangle := \int_{\bbG} \cdot \dd \xi.
\]
An $N$-th order moment system can be  written as 
\begin{equation}
  \label{eq:system1d}
   \pd{E_k}{t} + \pd{E_{k+1}}{x} = S_k,
   \quad 0\leq k\leq N,
 \end{equation}
where 
\[
   S_k := \langle \tau(\xi) v^k\mathcal{S}\rangle.
\]
However, the moment system \eqref{eq:system1d} is not closed, since
there are $N+1$ equations and $N+2$ variables
$E_0,E_1,\cdots,E_{N+1}$. Thus a moment closure needs to be applied. A
typical moment closure is to use a function of $E_0,E_1,\cdots,E_N$
instead of $E_{N+1}$ in \eqref{eq:system1d}, which is given by
\begin{equation}
 \label{eq:oldclosure1d}
  E_{N+1} = E_{N+1}(E_0,E_1,\cdots,E_N),
\end{equation}
and a commonly used method to obtain \eqref{eq:oldclosure1d} is to give
an ansatz $\hat{f}$, which satisfies  
\begin{equation}\label{eq:constrains_f}
  \langle \tau(\xi) v^k \hat{f}\rangle = E_k,\quad  0\leq k\leq N.
\end{equation}
Then $E_{N+1}$ is taken as the $(N+1)$-th moment of $\hat{f}$ to
finish the moment closure. Based on this idea, many practical moment
models
\cite{Grad,levermore1996moment,QuantumGrad13,alldredge2016approximating,MPN}
for kinetic equations have been proposed.

On the other hand, in order to make \eqref{eq:system1d} closed, one only needs
to give the moment closure on $\pd{E_{N+1}}{x}$. This inspires us to
take the moment closure as  
\begin{equation}
 \label{eq:newclosure1d}
 \pd{E_{N+1}}{x} = \pd{E_{N+1}}{x}
 \left(E_0,E_1,\cdots,E_N;\pd{E_0}{x},\pd{E_1}{x},\cdots,\pd{E_{N}}{x}\right).
\end{equation}
It is natural to expect a quasilinear system to be derived, that we
require that $\pd{E_{N+1}}{x}$ relies linearly on
$\pd{E_0}{x},\cdots,\pd{E_{N}}{x}$. Since we can use an ansatz
$\hat{f}$ to approximate the distribution function $f$, we can
similarly give an ansatz $\hat{g}$ on $\pd{f}{x}$, which satisfies
\begin{equation}
    \label{eq:constrains}
    \langle \tau(\xi)v^k \hat{g} \rangle = \pd{E_k}{x},\quad 0\leq k\leq N,
\end{equation}
and the moment closure \eqref{eq:newclosure1d} is then given by 
\[
  \pd{E_{N+1}}{x} = \langle \tau(\xi) v^{N+1}\hat{g} \rangle.
\]
\subsection{Model deduction}
\cite{Grad,QuantumGrad13,MPN} suggest using a weighted polynomial to
approximate the distribution function $f$ to give the moment closure
\eqref{eq:oldclosure1d}, which motivates us to use a
weighted polynomial to approximate $\pd{f}{x}$ to give the moment
closure \eqref{eq:newclosure1d}. Precisely, we take a positive 
weight function $\weight$, where $\bmeta$ is a vector composed of some parameters
depend on moments $E_0,E_1,\cdots,E_N$, and the ansatz is given by 
\begin{equation}
  \label{eq:ansatz}
  \hat{g} = \weight\sum_{i=0}^{N}g_iv^i.
\end{equation}
Denote $\mE_{i} = \langle \tau(\xi)v^i\weight\rangle$ as the $i$-th moment of
weight function $\weight$, and matrix $\bD\in\bbR^{(N+1)\times(N+1)}$
and $\bK\in\bbR^{(N+1)\times(N+1)}$ satisfy 
$\bD_{i,j} = \mE_{i+j}$, and $\bK_{i,j} = \mE_{i+j+1}$, $0\leq i,j\leq
N$, i.e.
\begin{equation}
  \label{eq:matrixDK}
  \bD = \left[
  \begin{array}[H]{ccccc}
    \mE_{0}&\mE_{1}&\mE_{2}&\cdots&\mE_{N}\\
    \mE_{1}&\mE_{2}&\mE_{3}&\cdots&\mE_{N+1}\\
    \mE_{2}&\mE_{3}&\mE_{4}&\cdots&\mE_{N+2}\\
    \vdots&\vdots&\vdots&\ddots&\vdots\\
    \mE_{N}&\mE_{N+1}&\mE_{N+2}&\cdots&\mE_{2N}
  \end{array}
  \right]
  ,\quad \bK = \left[
  \begin{array}[H]{ccccc}
    \mE_{1}&\mE_{2}&\mE_{3}&\cdots&\mE_{N+1}\\
    \mE_{2}&\mE_{3}&\mE_{4}&\cdots&\mE_{N+2}\\
    \mE_{3}&\mE_{4}&\mE_{5}&\cdots&\mE_{N+3}\\
    \vdots&\vdots&\vdots&\ddots&\vdots\\
    \mE_{N+1}&\mE_{N+2}&\mE_{N+3}&\cdots&\mE_{2N+1}
  \end{array}
  \right].
\end{equation}
Since $\tau(\xi)>0$ and $\weight(\xi)>0$, 
it is not difficult to prove that $\bD$ is symmetric and 
positive definite, and $\bK$ is symmetric. Furthermore,
denote $\bm g = (g_0,g_1,\cdots,g_N)^T\in\bbR^{N+1}$, then
according to the constraints \eqref{eq:constrains}, we have 
\begin{equation}
  \label{eq:constrains_matrix}
  \bD\bm g = \left( \pd{E_0}{x},\pd{E_1}{x},\cdots,\pd{E_N}{x}
  \right)^T,
\end{equation}
and the moment closure is given by 
\begin{equation}
  \label{eq:closure_matrix}
  \left( \pd{E_1}{x},\pd{E_2}{x},\cdots,\pd{E_{N+1}}{x}
  \right)^T  = \bK\bm g = \bK\bD^{-1}
  \left( \pd{E_0}{x},\pd{E_1}{x},\cdots,\pd{E_N}{x}
  \right)^T.
\end{equation}
Notice in \eqref{eq:closure_matrix}, the gradient of the flux
functions rely linearly on $\pd{E_k}{x}$, $k = 0, \cdots,N$, which
satisfies our previous requirement.  Therefore, the moment system
\eqref{eq:system1d} can be written as 
\begin{equation}
  \label{eq:system1d_matrix}
  \pd{\bm E}{t} + \bK\bD^{-1}\pd{\bm E}{x} = \bm S,
\end{equation}
where $\bm E = (E_0,E_1,\cdots,E_N)^T\in\bbR^{N+1}$, and $\bm S =
(S_0,S_1,\cdots,S_N)^T\in\bbR^{N+1}$.
\subsection{Hyperbolicity} 
The moment system \eqref{eq:system1d_matrix} is globally hyperbolic,
and the eigenvalues can be calculated.
We first give this theorem:
\begin{theorem}
  \label{thm:hyper1d}
  Moment system \eqref{eq:system1d_matrix} is globally hyperbolic.
\end{theorem}
\begin{proof}
  To obtain the hyperbolicity, one only needs to investigate whether
  $\bK\bD^{-1}$ is real diagonalizable. Notice that $\bD$ is symmetric
  and positive definite, and $\bK$ is symmetric, thus $\bK\bD^{-1}$ is
  real diagonalizable. Therefore, moment system
  \eqref{eq:system1d_matrix} is globally hyperbolic.
\end{proof}
\subsubsection{Orthogonal polynomials}
In order to investigate the eigenvalue of
\eqref{eq:system1d_matrix}, we first introduce a series of monic 
orthogonal polynomials $\pl_i(v)$, which satisfy
\begin{equation}
  \label{eq:orthogonality1d}
  \int_{\bbG} \tau(\xi)\weight(\xi)
  \pl_i(v)\pl_j(v)\dd\xi = 0,\quad \text{when } i \neq j.
\end{equation}
The Gram-Schmidt procedure to obtain the orthogonal polynomial can be formulated as 
\begin{equation}
  \label{eq:Gram1d}
  \pl_0 = 1,\quad \pl_i = v^i - \sum_{j=0}^{i-1}
  \dfrac{\mK_{i,j}}{\mK_{j,j}} \pl_j,
\end{equation}
where 
\begin{equation}
  \label{eq:mK1d}
  \mK_{i,j} := \langle \tau(\xi)\weight(\xi) v^i\pl_j\rangle =
  \int_{\bbG} \tau(\xi)\weight(\xi) v^i\pl_j\dd\xi.
\end{equation} 
According to the orthogonality \eqref{eq:orthogonality1d}, we have 
\[ 
 \mK_{i,i} = \langle \tau(\xi)\weight(\xi)(\pl_i)^2\rangle >0,\quad 
 \mK_{k,i} = 0, \quad \text{ when } k < i.
\]
Furthermore, multiplying $v^k\tau(\xi)\weight(\xi)$ by \eqref{eq:Gram1d},
and taking integration with respect to $\xi$ on $\bbG$, 
$\mK_{k,i}$ can be calculated by 
\[
\mK_{k,i} = \mE_{k+i} - \sum_{j=0}^{i-1}
\dfrac{\mK_{i,j}}{\mK_{j,j}}\mK_{k,j}.
\]
Therefore, if $\tau(\xi) > 0$ and $\weight(\xi)>0$, one can always
define the orthogonal polynomials $\pl_i$ by \eqref{eq:Gram1d}.
Then, on the eigenvalue of \eqref{eq:system1d_matrix}, we have the
following theorem: 
\begin{theorem}
  \label{thm:eigen1d}
  The eigenvalues of the moment system \eqref{eq:system1d_matrix} 
  are the zeros of the $(N+1)$-th orthogonal polynomial $\pl_{N+1}(v)$.
\end{theorem}
\begin{proof}
  Consider the characteristic polynomial of \eqref{eq:system1d_matrix} 
  \[
  \text{Det}(\bK\bD^{-1}-\lambda) = 0,
  \]
  which can be rewritten as  
  \[
  \text{Det}(\bK-\lambda\bD) = 0.
  \]
  Since $\lambda$ is a eigenvalue, there exists a vector
  $\bm r= (r_0,r_1,\dots,r_N)^T\in\bbR^{N+1}$, satisfying that 
  $(\bK-\lambda\bD)\br=0$, i.e.  
  \[
  \sum_{i=0}^{N}(K_{k,i}-\lambda D_{k,i})r_i = 0,\quad 0\leq k\leq N,
  \]
  which is 
  \[
  \sum_{i=0}^{N}(\mE_{k+i+1}-\lambda\mE_{k+i})r_i = 0 =
  \left\langle v^k(v-\lambda)\sum_{i=0}^{N}r_iv^i
  \tau(\xi)\weight(\xi)\right\rangle,\quad 0\leq k\leq N.
  \]
  Therefore, $(v-\lambda)\sum_{i=0}^{N}r_iv^i$ is orthogonal to
  $v^k$, with respect to $\tau(\xi)\weight(\xi)$, for any $0\leq k\leq N$. 
  Notice $(v-\lambda)\sum_{i=0}^{N}r_i v^i $ is a polynomial of $v$
  whose degree is not greater than $N+1$, 
  thus we have $r_N\neq 0$, and 
  \[
  (v-\lambda)\sum_{i=0}^{N} r_iv^i = r_N\pl_{N+1}.
  \]
 Therefore, $\lambda$ is the root of the orthogonal polynomial
  $\pl_{N+1}$.
\end{proof}
\subsection{Comparison with conventional moment system}
The conventional moment system, derived by postulating an ansatz for
the distribution function to obtain \eqref{eq:oldclosure1d}, can be
formulated as a conservation law, without considering the right hand
side. Due to the conservation, it is a popular way to deduce a
moment system. However, the conventional moment system often suffers
from lack of hyperbolicity. As a globally hyperbolic moment system,
in this section, we will show some properties of 
\eqref{eq:system1d_matrix} in comparison with the conventional moment system. 

First, it is direct to show the conservation of the moments with
orders from $0$ to $N-1$, i.e., we have
\begin{theorem}\label{thm:conservativeness}
In the moment system \eqref{eq:system1d}, the first $N$ equations 
can be written into conservation form. 
\end{theorem}
\begin{proof}
  The first $N$ equations, corresponding to the momensts with orders
  from $0$ to $(N-1)$, are
\[ 
\pd{E_k}{t} + \pd{E_{k+1}}{x} = S_k, \quad 0\leq k\leq N-1,
\]
therefore, the governing equations of the moments with orders from $0$
  to $(N-1)$ can be written into conservation form.
\end{proof}
However, compared to the conventional moment system given by \eqref{eq:oldclosure1d} and 
\eqref{eq:constrains_f}, the framework cannot ensure the 
conservation property of the entire system. Nevertheless, we have the following theorem:

\begin{theorem}\label{thm:special_case}
  For the one-dimensional kinetic equation \eqref{eq:kinetic1d}, if we
  use the ansatz 
  \begin{equation}\label{eq:ansatz_mn}
    \hat{f} = h\left(\sum\limits_{i=0}^N \eta_i v^i\right)
  \end{equation}
  to approximate the distribution function $f$ and derive a moment
  system $\mathcal{A}_1$. The resulting system is equivalent to the system \eqref{eq:system1d_matrix}
  by taking the weight function 
  \begin{equation}
    \label{eq:weight_mn}
    \weight = h'\left( \sum_{i=0}^{N}\eta_iv^i \right),
  \end{equation}
  which is denoted as $\mathcal{A}_2$.
  Moreover, the system \eqref{eq:system1d_matrix} in this situation is
  in conservation form. 
\end{theorem}
\begin{proof}
  We write the system derived by taking the ansatz
  \eqref{eq:ansatz_mn} into quasilinear form, formulated as 
  \begin{equation}
    \label{eq:system_mn}
    \pd{\bm{E}}{t} + \bM\pd{\bm{E}}{x} = \bm S,
  \end{equation}
  where the element in the $i$-th row and $j$-th column, which we
  hereafter will refer to as the $(i,j)$-th element of matrix $\bM$, is 
  $\bM_{i,j} = \pd{E_{i+1}}{E_{j}}$. Notice the moment closure
  $E_{N+1}$ is given by 
  \begin{equation}
     E_{N+1} = \langle \tau(\xi)v^{N+1}\hat{f}\rangle = 
     \left\langle \tau(\xi)v^{N+1}h\left( \sum_{i=0}^{N}\eta_iv^i 
     \right)\right\rangle,
  \end{equation}
  and the constraints provided by the first $N$ moments 
  \begin{equation}
     E_{k} = \langle \tau(\xi)v^{k}\hat{f}\rangle = 
     \left\langle \tau(\xi)v^{k}h\left( \sum_{i=0}^{N}\eta_iv^i 
     \right)\right\rangle, \quad 0\leq k\leq N,
  \end{equation}
  since $\tau(\xi)$ is independent 
  with $\eta_i$, thus when $0\leq i\leq N, 0\leq k\leq N+1$,  
  \[   
    \pd{E_k}{\eta_i} = \pd{\left\langle \tau(\xi)v^k
    h\left(\sum_{i=0}^{N}\eta_iv^i\right) \right\rangle}{\eta_i} = 
    \left\langle\tau(\xi)v^kh'\left( \sum_{i=0}^{N}\eta_iv^i
    \right)v^i\right\rangle = \left\langle \tau(\xi)\weight v^{i+k}\right\rangle
    =\mE_{k+i}.
  \]
  Therefore, $\pd{E_k}{\eta_i}$ is the $(k,i)$-th element in $\bD$
  \eqref{eq:matrixDK}, $0\leq i,k\leq N$, 
  and $\pd{\eta_i}{E_k}$ is the $(k,i)$-th
  element of $\bD^{-1}$, since $\bD$ is symmetric and positive
  definite. Furthermore, $\pd{E_{i+1}}{\eta_j}$ is $(i,j)$-th element
  of $\bK$, $0\leq i,j\leq N$. Thus, 
  \[
  \pd{E_{i+1}}{E_j} =
  \sum_{k=0}^{N}\pd{E_{i+1}}{\eta_k}\pd{\eta_k}{E_j} = \sum_{k=0}^{N}
  \bK_{i,k}\bD^{-1}_{k,j} = (\bK\bD^{-1})_{i,j}.
  \]
  Therefore, $\bM= \bK\bD^{-1}$, and the
  system \eqref{eq:system_mn} is equivalent to \eqref{eq:system1d_matrix}.
  
  At last, since \eqref{eq:system_mn} can be written as 
  \[
  \left\langle \tau(\xi)v^k \left(\pd{\hat{f}}{t} +
  v\pd{\hat{f}}{x} \right)
  \right\rangle =
  \left\langle \tau(\xi)v^k\mS \right\rangle, \quad 0\leq k\leq N,
  \]
  \eqref{eq:system_mn} and \eqref{eq:system1d_matrix} are in
  conservation form. 
\end{proof}
Furthermore, when the number of parameters $\bm\eta$ is less than $N$,
we have the following theorem
\begin{theorem}
  \label{thm:maxwelliter}
  For the one-dimensional kinetic equation \eqref{eq:kinetic1d}, if we
  use the weight function 
  \begin{equation}\label{eq:ansatz_mnpn}
     \weightf = h\left(\sum\limits_{i=0}^n \eta_i v^i\right),
  \end{equation}
  and make a weighted polynomial $\hat{f} = \sum_{i=0}^{N}f_iv^i\weightf$
  to approximate the distribution function $f$ and derive a moment
  system $\mathcal{A}_1$, and denote the system
  \eqref{eq:system1d_matrix} by taking the weight function 
  \begin{equation}
    \label{eq:weight_mnpn}
    \weight = h'\left( \sum_{i=0}^{n}\eta_iv^i \right),
  \end{equation}
  as $\mathcal{A}_2$.  Then the results of one-step Maxwellian
  iteration of $\mathcal{A}_1$ and $\mathcal{A}_2$ are the same. 
\end{theorem}
\begin{proof}
  We still use the notation in \eqref{eq:matrixDK}, and
  $\mathcal{A}_2$ can be written as 
  \[  
  \pd{\bm{E}}{t} + \bK\bD^{-1}\pd{\bm{E}}{x} = \bm{S}.
  \]
  On the other hand, denote $\mEn_n=\langle \tau(\xi)v^{n}\weightf\rangle$, and $
  \bDn$ as $\bDn_{i,j} = \mEn_{i+j}$,
  and $\bKn$ as 
  $\bKn_{i,j} = \mEn_{i+j+1}$, $0\leq i,j\leq N$, then
  $\mathcal{A}_1$ can be written as 
  \[
  \pd{\bm{E}}{t} + \pd{(\bKn\bDn^{-1}\bm{E})}{x} = \bm{S}.
  \]
  Furthermore, we can rewrite these two systems by a variable
  $\bw=(\bm\eta,\bm\eta^*)\in\bbR^{N+1}$, where there is an one-to-one map
  between $\bw$ and $\bm{E}$. Additionally, we assume that when
  $\hat{f}=\weightf$, $E_k=\mEn_k$, and  $\bm\eta^*=0$. 
  Without loss of generality, we assume that 
  $\bA := \pd{\bm E}{\bw} $ is invertible. 

  Then $\mathcal{A}_1$ and $\mathcal{A}_2$ can be respectively written
  as 
  \[
    \bA\pd{\bw}{t} + \pd{(\bKn\bDn^{-1}\bm E)}{\bw}\pd{\bw}{x} = \bm{S},\quad 
    \bA\pd{\bw}{t} + \bK\bD^{-1}\bA\pd{\bw}{x} = \bm{S},\quad 
  \]
  In order to prove that the one-step Maxwellian iteration of these
  two systems are equivalent, one only needs to prove that when
  $\hat{f} = \weightf$, i.e. $E_k = \mEn_k$ for $0\leq k\leq N$,  
  \begin{equation} 
    \label{eq:thm_maxwell_target}
    (\bK\bD^{-1}\bA)_{m,l}\Big|_{\bm\eta^*=0} = \pd{(\bKn\bDn^{-1}\bm{E})_m}{\eta_l}
    \Big|_{\bm\eta^*=0},\quad
  0\leq m\leq N,
  0\leq l\leq n.
\end{equation}

  The right hand side can be calculated as 
  \begin{equation}
    \label{eq:linear_cal1}
  \begin{aligned}
  \pd{(\sum_{i,j=0}^{N}\bKn_{m,i}\bDn^{-1}_{i,j}E_j)}{\eta_l} \Big|_{\bm\eta^*=0}
  =& 
  \pd{[(\sum_{i,j=0}^{N}\bKn_{m,i}\bDn^{-1}_{i,j}E_j)\big|_{\bm\eta^*=0}]}{\eta_l} 
  = \pd{\bKn_{m,0}}{\eta_l} = \bK_{m,l}
  \end{aligned} 
\end{equation}
  where when $\hat{f}=\weightf$, $E_j=\mEn_{j}=\bDn_{j,0}$, thus
  \[
  \left(\sum_{j=0}^{N}\bDn^{-1}_{i,j}E_j\right)\Big|_{\bm\eta^*=0} =
  \begin{cases}
    0,& i\neq 0,\\
    1, & i =0,
  \end{cases}
  \]
  and 
  \[  
     \pd{\bKn_{i,0}}{\eta_l} = \bK_{i,l},
     \quad 
     \pd{\bDn_{i,0}}{\eta_l} = \bD_{i,l},
     \quad 0\leq i\leq N,0\leq l\leq n.
  \]
  are applied. On the other hand, notice that 
  \[
  \bA_{i,l}|_{\bm\eta^*=0} = \pd{E_i}{\eta_l}\Big|_{\bm\eta^*=0} = 
  \pd{\bDn_{i,0}}{\eta_l}\Big|_{\bm\eta^*=0} = \bD_{i,l}|_{\bm\eta^*=0}. 
  \]
  Therefore, 
  \begin{equation} 
    (\bK\bD^{-1}\bA)_{m,l}\Big|_{\bm\eta^*=0} = \bK_{m,l}, 
  0\leq m\leq N,
  0\leq l\leq n,
\end{equation}
  thus \eqref{eq:thm_maxwell_target} holds.
\end{proof}
\begin{remark}
  The result of Maxwellian iteration for many kinetic equations and their 
  moment models is important. For instance, the one-step
  Maxwellian iteration of Grad's 13-moment equation leads to
  Navier-Stokes-Fourier law \cite{Struchtrup,di2016quantum}. However, as pointed 
  out in \cite{di2016quantum, HMPN}, a direct application of the hyperbolic regularization in 
  \cite{Fan2015, framework} may change the result of Maxwellian iteration. According to 
  above theorem, the moment system \eqref{eq:system1d_matrix} fixes this defect.
\end{remark}


%% file: example1d.tex

\section{Applications in one-dimensional case}\label{sec:example1d}
In this section, we first list some existing hyperbolic moment models, some
of which can be regarded as special cases of our framework, i.e., they
can be written as the form in 
\eqref{eq:system1d_matrix} by taking a
proper weight function $\weight$. Then we will 
also propose some novel hyperbolic moment models based on our new framework.

In order to put moment models into our framework, we need to prove
they are equivalent to \eqref{eq:system1d_matrix}.
Apart from \cref{thm:special_case}, 
in order to judge whether two moment models are the same, we introduce this lemma. 
\begin{lemma}
  \label{lem:equality}
  If two one-dimensional $N$-th order moment system with $N+1$ equations, denoted as
  system $\mathcal{A}_1$ and $\mathcal{A}_2$, satisfy these two conditions:
  \begin{enumerate}
    \item 
  For each system, the governing equations of the moments with orders
      from $0$ to $(N-1)$ can be written in conservation form. In other words, 
  the first $N$ equations are equivalent to 
  \[
    \pd{E_k}{t} + \pd{E_{k+1}}{x} = S_k,\quad 0\leq k\leq N-1.
  \]
\item The characteristic polynomial of $\mathcal{A}_1$ is the same as
  the characteristic polynomial of $\mathcal{A}_2$.
  \end{enumerate}
  Then these two moment models are equivalent. 
\end{lemma}
\begin{proof}
  According to the first condition, the moment system can be written
  as these quasi-linear forms
  \[ 
  \pd{\bm{E}}{t} + {\bf A}_1\pd{\bm{E}}{x} = \bm{S}
  \text{  and  }\quad  
  \pd{\bm{E}}{t} + {\bf A}_2\pd{\bm{E}}{x} = \bm{S},
  \]
  where $\bm{E} = (E_0,E_1,\cdots,E_N)^T\in\bbR^{N+1}$, and $\bm{S} =
  (S_0,S_1,\cdots,S_N)^T\in\bbR^{N+1}$. The coefficient matrix ${\bf
  A}_1$ and ${\bf A}_2$ can be written as 
  \[ 
    {\bf A}_1 = \left[
    \begin{array}[H]{cccccc}
      0&1&0&0&\cdots&0\\
      0&0&1&0&\cdots&0\\
      0&0&0&1&\cdots&0\\
      \vdots&\vdots&\vdots&\vdots&\ddots&\vdots\\
      0&0&0&0&\cdots&1\\
      c_{1,0}&c_{1,1}&c_{1,2}&c_{1,3}&\cdots&c_{1,N}
    \end{array}
    \right]
    \text{ and }\quad 
    {\bf A}_2 = \left[
    \begin{array}[H]{cccccc}
      0&1&0&0&\cdots&0\\
      0&0&1&0&\cdots&0\\
      0&0&0&1&\cdots&0\\
      \vdots&\vdots&\vdots&\vdots&\ddots&\vdots\\
      0&0&0&0&\cdots&1\\
      c_{2,0}&c_{2,1}&c_{2,2}&c_{2,3}&\cdots&c_{2,N}
    \end{array}
    \right].
  \]
  Therefore, the characteristic polynomials are 
  \[
  p_1(\lambda) = \lambda^{N+1} - \sum_{i=0}^{N} c_{1,i}\lambda^i
  \text{ and }\quad 
  p_2(\lambda) = \lambda^{N+1} - \sum_{i=0}^{N} c_{2,i}\lambda^i.
  \]
  Noticing that $p_1(\lambda)=p_2(\lambda)$, we have
  $c_{1,i}=c_{2,i}$,$i=0,1,\cdots,N$. Thus we have 
  \[
  {\bf A}_1 = {\bf A}_2,
  \]
  thus $\mathcal{A}_1$ and $\mathcal{A}_2$ are equivalent.
\end{proof}

\subsection{Boltzmann equation} \label{sec:Boltzmann1d}
This section considers the one-dimensional Boltzmann equation. The
Boltzmann equation depicts the movement of the gas molecules from a
statistical point of view \cite{Boltzmann,Struchtrup}. It reads  
\begin{equation}
  \label{eq:Boltzmann1d}
  \pd{f}{t} + \xi \pd{f}{x} = \mS(f), 
\end{equation}
where $f=f(t,x,\xi)$ is the distribution function, 
$\xi\in \bbR$ is the microscopic velocity, and the right-hand
side $\mS$ is the collision term, 
depicting the interaction between gas particles.  

The thermodynamic equilibrium is 
\begin{equation}\label{eq:boltzmann1d_feq}
  f_{eq} = \dfrac{\rho}{\sqrt{2\pi\theta}}\exp\left(
  -\dfrac{|\xi-u|^2}{2\theta} \right), 
\end{equation}
where $\rho$, $u$ and $\theta$ are macroscopic variables, satisfying
that $\rho = \langle f\rangle$, $\rho u=\langle \xi f\rangle$ and
$\rho\theta=\langle |\xi-u|^2 f\rangle$.

Furthermore, the commonly used definition of the moments are 
\[  
  E_k = \langle \xi^kf\rangle, \quad k\geq 0.
\]
Compared with \eqref{eq:kinetic1d}, we have $v(\xi)=\xi$ and
$\tau(\xi)=1$.

\subsubsection{One-dimensional HME model} \label{sec:HME1d}
To derive a reduced model from \eqref{eq:Boltzmann1d},
Grad \cite{Grad} proposed the moment method, in which a weighted
polynomial is applied to approximate the distribution function, with
the weight function set as $f_{eq}$, i.e. 
\begin{equation}\label{eq:weight_hme}
  \weightf = \dfrac{\rho}{\sqrt{2 \pi \theta}}
  \exp\left(-\frac1{2\theta}(\xi-u)^2\right).
\end{equation}
  By taking 
  $\bm\eta=\left(\ln\left(\dfrac{\rho}{\sqrt{2\pi\theta}}\right)
-\dfrac{u^2}{2\theta}, \dfrac{u}{\theta},-\dfrac{1}{2\theta}\right)^T$, 
  The weight function \eqref{eq:weight_hme} can be written as the form
  \eqref{eq:ansatz_mnpn}, with $h(\zeta)=\exp(\zeta)$. 

However, in \cite{Muller,Grad13toR13}, 
it was pointed out the Grad's
moment method is not globally hyperbolic, which limits the
applications of the Grad's moment method. Therefore,
the hyperbolic moment equation \cite{Fan,Fan_new} is proposed as 
the result of a globally hyperbolic
regularization of Grad's moment system.

With respect to the weight function \eqref{eq:weight_hme}, the
orthogonal polynomial is the generalized Hermite polynomial,
formulated as 
\begin{equation}
  \label{eq:He1d}
\He_{k}(\xi) = (-1)^k
\exp\left( -\frac{(\xi-u)^2}{2\theta} \right)
\od{^k}{\xi^k} 
\exp\left( -\frac{(\xi-u)^2}{2\theta} \right),
\end{equation}
it is not difficult to check that 
\[  
\He_k(\xi) = \theta^{-k/2}\text{He}_k\left( \dfrac{\xi-u}{\sqrt{\theta}}
\right),
\]
where $\text{He}_k$ is the $k$-th order Hermite polynomial. Therefore,
the zeros of $\He_k$ are  $u+c_i\sqrt{\theta}$, where $c_i$ are zeros
of Hermite polynomial $\text{He}_k$, $1\leq i\leq k$. In other words,
when we take the weight function as \eqref{eq:weight_hme} in our
framework, the eigenvalues of the moment system
\eqref{eq:system1d_matrix} are the zeros of $\He_{N+1}$. 

On the other hand,
according to the discussions in \cite{Fan_new}, the governing equation
of the moments with orders from $0$ to $N-1$ (i.e. $E_{k}$, $k\leq N-1$) in HME system can
be written in conservation form, and the eigenvalues of the HME
system is the zeros of $\He_{N+1}$.  According to
\cref{lem:equality}, HME system can be put into our framework, by
taking the weight function as 
\[  
   \weight = h'\left(\sum_{i=0}^{n}\eta_i\xi^i\right)
   =\exp\left(\sum_{i=0}^{n}\eta_i\xi^i\right)
   = \dfrac{\rho}{\sqrt{2 \pi \theta}}
  \exp\left(-\frac1{2\theta}(\xi-u)^2\right).
\]

\subsubsection{Maximum entropy model}
The maximum entropy model \cite{levermore1996moment} 
makes use of the following ansatz for the
distribution function:
\begin{equation}\label{eq:mep_ansatz}
  \hat{f}(t,x,\xi) = \exp\left(\sum\limits_{i=0}^N \eta_i(t,x) \xi^i
  \right),
\end{equation}
which corresponds to $h(\zeta) = \exp(\zeta)$ in \cref{thm:special_case}.
Therefore, according to \cref{thm:special_case}, take the weight
function as 
\[  
\weight = h'\left( \sum_{i=0}^{N} \eta_i\xi^i \right) = 
\exp\left( \sum_{i=0}^{N}\eta_i\xi^i \right),
\]
and we know the two moment models are the same.

\subsection{Radiative transfer equation} \label{sec:RTE1d}
This section considers the radiative transfer equation (RTE) in slab geometry 
\begin{equation}
  \dfrac{1}{c}\pd{f}{t} + \xi \pd{f}{x} = \mS(f, T),
\end{equation}
where $c$ is the speed of light, $\xi \in [-1,1]$ denotes the cosine
of the angle between the photon velocity direction and the $x$ axis,
and the right-hand side $\mS$ \cite{Bru02, McClarren2008Semi} 
depicts the interaction between photons
and the background medium which has material temperature $T$.

Furthermore, the commonly used definition of the moments are 
\[  
  E_k = \langle \xi^kf\rangle, \quad k\geq 0.
\]
Compared with \eqref{eq:kinetic1d}, we have $v(\xi)=\xi$ and
$\tau(\xi)=1$. 

\subsubsection{\PN model}
The \PN model \cite{jeans1917stars} 
suggests to approximate the distribution function by a polynomial,
i.e.
\begin{equation}
  \hat{f}(t,x,\xi) = \sum\limits_{i = 0}^N
  \eta_i(t,x)\xi^i,
\end{equation}
which corresponds to \cref{thm:special_case} with $h(\zeta) = \zeta$.
Therefore, take the weight function as 
\[
\weight = 1,
\]
then according to \cref{thm:special_case}, we know that the moment
system derived by our framework is the same as the \PN model.

\subsubsection{\MN model}\label{sec:MNofRTE1d}
To derive the \MN model \cite{minerbo1978maximum,levermore1996moment,
dubroca1999theoretical}
for gray RTE, one uses the following ansatz for the
distribution function:
\begin{equation}\label{eq:mn_ansatz}
  \hat{f}(t,x,\xi) = \left[\sum\limits_{i=0}^N \eta_i(t,x) \xi^i
  \right]^{-4},
\end{equation}
which corresponds to $h(\zeta) = \dfrac{1}{\zeta^4}$ in
\cref{thm:special_case}.

On the other hand, for the \MN model for monochromatic radiative transfer, 
the ansatz to approximate the distribution function is given by 
\begin{equation}
  \hat{f}(t,x,\xi) = \left[\exp\left(\sum\limits_{i=0}^N \eta_i(t,x) \xi^i\right) -
  1\right]^{-1},
\end{equation}
which corresponds to $h(\zeta) = \dfrac{1}{\exp(\zeta)-1}$ in
\cref{thm:special_case}.

Therefore, according to \cref{thm:special_case}, we take the weight function 
for the grey RTE as 
\[ 
  \weight = \dfrac{1}{\left(\sum_{i=0}^{N}\eta_i\xi^i\right)^5},
\]
and the weight function for the monochromatic RTE is 
\[
\weight = \dfrac{\exp\left(\sum_{i=0}^{N}\eta_i\xi^i\right)}{
\left(\exp\left(\sum_{i=0}^{N}\eta_i\xi^i \right)-1\right)^2}.
\]
Then the \MN model for the grey RTE and the monochromatic RTE can also
regarded as special cases of our new framework
\eqref{eq:system1d_matrix}.

\subsubsection{\HMPN model}
In \cite{MPN}, the researchers suggest to use a weighted polynomial
to approximate the specific intensity $f$, where the weight function
is given by the ansatz of the \Mone model in grey medium, i.e. 
\[ 
 \hat{f} = \dfrac{1}{\left( \eta_0+\eta_1\xi \right)^4}
\sum_{i=0}^{N}f_i\xi^i,
\]
where $\eta_0,\eta_1$ are determined by the moments
$E_0$ and $E_1$, and 
$\dfrac{\eta_1}{\eta_0}\in(-1,1)$. 

Furthermore, in \cite{HMPN}, it was pointed out that the \MPN model is
not globally hyperbolic for $N\geq 3$, thus a hyperbolic
regularization is needed. 
Direct application of the framework proposed in \cite{framework,
Fan2015} results in a globally hyperbolic moment system. However, this
system has other physical defects. When $N=1$, this system is
different from the \Mone model. Moreover, this system's 
one-step Maxwellian changes the result of the \MPN model. Therefore,
in \cite{HMPN} a new hyperbolic regularization was proposed to fix these
defects, and the resulting system is \HMPN model. 
Now we claim that the \HMPN model can also be regarded as an
example of our new framework. 

Noticing that in \MPN model, we use $\weightf = 1/(\eta_0+\eta_1\xi)^4$ as the weight
function to approximate the specific intensity $f$. Thus it is natural
to use its derivative to approximate $\pd{f}{x}$, i.e. the weight
function is taken as 
\begin{equation}
  \label{eq:weight_hmpn1d}
\weight = \dfrac{1}{\left(\eta_0+\eta_1\xi \right)^5}.
\end{equation}
Actually, the weight function is the same as the weight function we used in the \MN model in
\cref{sec:MNofRTE1d}, when $N = 1$.

According to \cref{lem:equality}, one only need to check 
\begin{enumerate}
  \item According to \cite{HMPN}, the governing equations of the
    moments with orders from $0$ to $N-1$ can be written into conservative form.
  \item According to \cite{HMPN}, the characteristic speed of the
    \HMPN model is the zeros of the $(N+1)$-th orthogonal polynomial
    with respect to the weight function \eqref{eq:weight_hmpn1d}.
\end{enumerate}
Therefore, we have the \HMPN model \cite{HMPN} is also a special case
of our framework, when the weight function is taken as
\eqref{eq:weight_hmpn1d}.
According to \cref{thm:special_case}, we have that the moment
system \eqref{eq:system1d_matrix} is the same as the \Mone model. 
According to \cref{thm:maxwelliter}, the Maxwellian iteration of 
the \HMPN model is the same as the \MPN model.

\subsubsection{\HMPN model for monochromatic case}
Since the \HMPN model for grey medium leads to a globally hyperbolic
models which also preserves nice physical properties \cite{HMPN}, it is natural to propose
the \HMPN model by taking the weight function as the monochromatic
\Mone model, i.e. 
\begin{equation}
  \weightf = \dfrac{1}{\exp(\eta_0+\eta_1\xi)-1},
\end{equation}
which corresponds to \eqref{eq:weight_mnpn} with $h(\zeta) =
\frac{1}{\exp(\zeta)-1}$. 

Taking the weight function as 
\begin{equation}
  \weight =
  \dfrac{\exp(\eta_0+\eta_1\xi)}{(\exp(\eta_0+\eta_1\xi)-1)^2},
\end{equation}
one can propose a globally hyperbolic
moment system with the first $(N-1)$-order moments in conservation
form. Moreover,
according to \cref{thm:eigen1d}, we have that the characteristic
speeds of the moment system are not greater than the speed of light.
Moreover, according to \cref{thm:special_case}, when $N=1$, the
resulting system is the \Mone model. According to
\cref{thm:maxwelliter}, the resulting system preserves the result of
the Maxwellian iteration the \MPN model for monochromatic case.

\subsection{Boltzmann-Peierls equation}
The one-dimensional Boltzmann-Peierls equation \cite{peierls1997kinetic}, which characterizes
phonon transport, reads
\begin{equation}
  \pd{f}{t} + v(\xi) \pd{f}{x} = \mS(f),
\end{equation}
where $v(\xi) = \od{w}{\xi}$, $w(\xi)$ is the dispersion relationship.
For instance, if we only consider harmonic interaction between
adjacent atoms \cite{guo2015phonon}, 
\begin{equation} w(\xi) = 2
  \sqrt{\dfrac{K}{m}}\left|\sin\dfrac{\xi a}{2}\right|, 
\end{equation}
where $\xi$ is the wave number of lattice vibration, while $K$, $m$
and $a$ are constants. Usually, $\xi
\in \mB = [-\pi/a, \pi/a)$ is from the first Brillouin zone. $\mS(f)$
depicts the phonon scattering process.

As a high order extension of the approach in \cite{banach2004nine}, 
the moments are defined as 
\[  
  E_k = \langle w v^kf\rangle,
\]
i.e. $\tau(\xi) = w(\xi)$ in this case. $E_0$ is proportional to
the energy density, and $E_1$ proportional to the heat flux.
We consider only the single mode relaxation time approximation
\cite{peraud2011efficient}, then equilibrium distribution function is
\begin{equation}
  f_{eq} = \dfrac{1}{\exp(\hbar w(\xi) / (k_B T)) - 1},
\end{equation}
where $\hbar$ is the reduced Planck constant, $k_B$ is the Boltzmann constant, and $T$ is 
temperature defined by energy conservation
\begin{equation}
  \left\langle \hbar w \dfrac{f_{eq} - f}{\tilde{\tau}} \right\rangle = 0,
\end{equation}
where $\tilde{\tau}$ is relaxation time.

A natural idea is to use a weighted polynomial to approximate the
distribution function $f$, with the weight function taken as the
equilibrium $f_{eq}$. Some recent works can be found in \cite{banach2004nine,
banach2004nine2}.

Using the conventional moment method, the weight function is taken as 
\[  
\weightf = \dfrac{1}{\exp\left( \eta_0 w(\xi) \right) - 1},
\]
and the ansatz is given by 
\[
  \hat{f} = \weightf \sum_{i=0}^{N} f_i v^i.
\]
The resulting system is then formulated as 
\begin{equation}\label{eq:trad_bp_1d} 
  \pd{\bm{E}}{t} + \bM\pd{\bm{E}}{x} = \bm{S},
\end{equation}
with $\bM_{i,j} = \pd{E_{i+1}}{E_{j}}$, and $E_{N+1}$ is given by 
\[ 
  E_{N+1} = \langle w(\xi)\hat{f}v^{N+1}\rangle. 
\]

Under our new framework, the derivative $\pd{f}{x}$ can be
approximated by a weighted polynomial, with 
\begin{equation}
  \weight = [\exp(\eta_0 w(\xi)) - 1]^{-2} \exp(\eta_0w(\xi)) w(\xi).
\end{equation}
Precisely, the ansatz is formulated as 
\begin{equation}
  \pd{f}{x} \approx \hat{g} = \weight \sum\limits_{i=0}^N g_i v^i.
\end{equation}
According to \cref{thm:hyper1d}, the moment system is globally hyperbolic. 
On the Maxwellian iteration, notice that 
\[ 
\pd{\mEn_i}{\eta_0} = 
\pd{\langle w(\xi)\weightf v^i \rangle}{\eta_0} =
\langle  w(\xi)\pd{\weightf}{\eta_0} v^i\rangle 
= \langle w(\xi)\weight v^i\rangle = \mE_i,
\]
and according to proof of \cref{thm:maxwelliter}, the moment system of
the new framework has the same result of Maxwellian iteration with the
conventional moment system \eqref{eq:trad_bp_1d}. 


%% file: modelmultid.tex

\section{Model in multi-dimensional case}\label{sec:modelmultid}
This section considers the D-dimensional kinetic equation, which has the form
\begin{equation}
  \label{eq:kineticmultid}
  \pd{f}{t} + \bv(\bxi)\cdot\nabla_{\bx}f = \mathcal{S}(f),\quad
  t\in\bbR^+,\bx\in\bbR^D,\bv(\bxi)\in\bbR^D,\bxi\in\bbG\subset\bbR^D,
\end{equation}
where the distribution function $f=f(t,\bx,\bxi)$, and
$\bv(\bxi)\in\bbR^D$ is a function of $\bxi$. 
The $k$-th moment is defined by 
\begin{equation}\label{eq:moment_multid}
 E_{k}(t,\bx):=\langle\tau(\bxi)\psi_{k}(\bv) f \rangle,\quad  0\leq k\leq
 M-1,
 \end{equation}
where $\psi_k$ are chosen polynomials of $v$, $M$ is the number of
moments considered in the moment system, and $\tau(\bxi)>0$.

In different moment models, the basis function $\psi_k$ could be
different. For example, in Grad's 13-moment model,
$D=3$, $M=13$,
$\tau(\bxi)=1$, $\bv(\bxi)=\bxi$,  
and the sequence of $\psi_k$ is 
\begin{equation}  \label{eq:base_13m}
\bm{\psi} =
(1,\xi_1,\xi_2,\xi_3,\xi_1^2,\xi_1\xi_2,\xi_1\xi_3,\xi_2^2,\xi_2\xi_3,\xi_3^2,
\Vert\bxi\Vert^2\xi_1,
\Vert\bxi\Vert^2\xi_2,
\Vert\bxi\Vert^2\xi_3)^T.
\end{equation}
In the $N$-th order HME model \cite{Fan_new}, we have $M=\binom{N+D}{D}$,
$\tau(\bxi)=1$, $\bv(\bxi)=\bxi$, and
$\psi_{\mN(\alpha)} = \bxi^{\alpha}$, where $\bxi^{\alpha} =
\prod_{d=1}^{D}\xi_d^{\alpha_d}$, and
$\mN(\cdot)$ is a map
from the multi-index group \( \{\alpha\in\bbN^D:|\alpha|\leq N \} \)
to the set $\{0,1,2,\cdots,M-1\}$. 

As $\langle \tau(\bxi)f\rangle$ often corresponds to density,
hereafter we assume that $\psi_0=1$. 

A moment system can be written as 
\begin{equation}
  \label{eq:systemmulti-d}
   \pd{E_{k}}{t} + \sum\limits_{d=1}^D
   \pd{\langle\tau(\bxi)v_d\psi_k f\rangle}{x_d}
   = S_{k},
   \quad 0\leq k\leq M-1,
 \end{equation}
where 
\[
  S_{k} := \langle \tau(\bxi) \psi_k\mathcal{S}\rangle, 
\]
It is obvious that the moment
system \eqref{eq:systemmulti-d} is not closed. Thus one has to seek a
moment closure. Similar as the 1-D case, we define the
moment closure
$\pd{\langle\tau(\bxi)v_d\psi_k f\rangle}{x_d}$, $d=1,2,\cdots,D$,
by imposing an ansatz on each $\pd{f}{x_d}$.

\subsection{Model deduction}
Similar to 1-D case, we approximate $\pd{f}{x_d}$ by a weighted
polynomial $\hat{g}_d$, $d=1,2,\cdots,D$.
We take a positive weight function $\weight$, where $\bm\eta$ depends
on $E_{k}$, $0\leq k\leq M-1$. Thus for different directions $x_d$,
their weight function is the same.
The ansatz is then given by 
\begin{equation}
  \label{eq:ansatz_multi}
  \hat{g}_d = \weight\sum_{i=0}^{M-1}g_{d,i}\psi_i,\quad
  1\leq d\leq D.
\end{equation}

Denote $\mE_{k} = \langle
\tau(\bxi)\psi_k\weight\rangle$ is the $k$-th moment of
weight function $\weight$. Let matrix $\bD\in\bbR^{M\times M}$
and $\bK_d \in\bbR^{M \times M}$, satisfying that 
$\bD_{i,j} = \langle\tau(\bxi)\psi_i\psi_j\weight\rangle$, and
$\bK_{d, i, j} = \langle \tau(\bxi)v_d\psi_i\psi_j \weight\rangle$,
$0\leq i,j\leq M-1$.

Similar to 1-D case, since $\tau(\bxi)\weight>0$, we have $\bD$ is symmetric and 
positive definite, and $\bK_d$ is symmetric for all $d = 1, \cdots, D$. Furthermore,
for a given $1\leq d\leq D$, denote $\bm g_d \in \bbR^M$ to satisfy that its $k$-th
element is $g_{d, k}$, $0\leq k\leq M-1$, then according to the
constraints 
\begin{equation}
  \langle \tau(\bxi) \psi_k \hat{g}_d \rangle =
  \pd{E_{k}}{x_d},\quad 0\leq k\leq M-1,
\end{equation}
we have 
\begin{equation}
  \label{eq:multid_constrains_matrix}
  \bD\bm g_d = \left( \pd{E_0}{x_d},\pd{E_1}{x_d},\cdots,\pd{E_{M-1}}{x_d}
  \right)^T,
\end{equation}
and the moment closure is given by 
\begin{equation}
  \label{eq:closure_matrix_multiD}
  \left( 
  \pd{\langle \tau(\bxi)v_d\psi_0f\rangle}{x_d},\cdots,
  \pd{\langle \tau(\bxi)v_d\psi_{M-1}f\rangle}{x_d}
  \right)^T  = \bK_d\bm g_d = \bK_d\bD^{-1}
  \left(
  \pd{E_0}{x_d},\pd{E_{1}}{x_d},\cdots,\pd{E_{M-1}}{x_d}
  \right)^T.
\end{equation}
Therefore, the multi-dimensional moment system can be written as 
\begin{equation}
  \label{eq:systemmulti_d_matrix}
  \pd{\bm E}{t} + \sum\limits_{d=1}^D \bK_d \bD^{-1}\pd{\bm E}{x_d} = \bm S,
\end{equation}
where $\bm E = (E_{0},E_{1},\cdots,E_{M-1})^T\in\bbR^{M}$, and $\bm S =
(S_0,S_{1},\cdots,S_{M-1})^T\in\bbR^{M}$.
\subsection{Hyperbolicity} 
\begin{theorem}
  \label{thm:hypermultid}
  The moment system \eqref{eq:systemmulti_d_matrix} is globally hyperbolic.
\end{theorem}
\begin{proof}
  Notice that $\tau(\bxi)\weight>0$ implies that 
  $\bD$ is symmetric and positive definite and $\bK_d$ is
  symmetric for $d = 1, \cdots, D$, thus any linear combination
  $\sum\limits_{d=1}^D n_d \bK_d \bD^{-1}$ is real diagonalizable. Therefore,
moment system \eqref{eq:systemmulti_d_matrix} is globally hyperbolic.
\end{proof}

\subsection{Comparison with the conventional moment system}
It is easy to obtain this theorem 
\begin{theorem}\label{thm:conservation_md}
In the multi-dimensional moment system, the moment  $\langle \tau(\bxi)\phi f\rangle$ 
satisfy an equation in conservation form if $v_d\phi$ can be linearly expressed by
$\psi_l,0\leq l\leq M-1$, for any
$d=1,2,\cdots,D$. 
\end{theorem}
\begin{remark}
  When the considered moments are the moments with orders from $0$ to $N$, i.e.
  $\psi_{\mathcal{N}(\alpha)} = \bv^{\alpha}$, then the moments with
  orders from $0$ to  
  $(N-1)$, $\langle \tau(\bxi)\bxi^{\alpha} f\rangle$,
  $|\alpha|\leq N-1$ are conservative. 
\end{remark}
\begin{remark}
  In 13-moment system, i.e. $\psi_k$ is given by \eqref{eq:base_13m},
  then 
  \[ 
    \langle \tau(\bxi)f\rangle, \quad \langle \tau(\bxi)\bxi f\rangle,
    \quad \langle \tau(\bxi)|\bxi|^2 f\rangle,
  \]
  are conservative, which often correspond to density $\rho$, momentum
  $\rho \bm u$, and second-order moment $\rho\theta + \rho |\bm u|^2$.
\end{remark}
However, for the entire system, the conservation
form is not always preserved.
The following theorem discusses a special case.

\begin{theorem}\label{thm:special_case_multiD}
  For the multi-dimensional kinetic equation \eqref{eq:kineticmultid}, if we
  use the ansatz 
  \begin{equation}\label{eq:ansatz_mn_multiD}
    \hat{f} = h\left(\sum\limits_{k=0}^{M-1} \eta_{k}
    \psi_k\right)
  \end{equation}
  to approximate the distribution function $f$ and derive a moment
  system $\mathcal{A}_1$. The resulting system is equivalent to the system
  \eqref{eq:systemmulti_d_matrix} 
  by taking the weight function 
  \begin{equation}
    \label{eq:weight_mn_multiD}
    \weight = h'\left(
    \sum_{k=0}^{M-1}\eta_{k}\psi_k \right),
  \end{equation}
  which is denoted as $\mathcal{A}_2$.
  Moreover, the system \eqref{eq:systemmulti_d_matrix} in this situation is
  in conservation form.
\end{theorem}
\begin{proof}
  We rewrite the moment system $\mathcal{A}_1$ into quasi-linear
  form, 
  \begin{equation}
    \label{eq:system_mn_multiD}
    \pd{\bm{E}}{t} + \sum\limits_{d=1}^D \bM_d \pd{\bm{E}}{x_d} = \bm S,
  \end{equation}
  where the $(i,j)$-th element of matrix
  $\bM_d$ is 
  $\bM_{d,i,j} = \pd{ \langle\tau(\bxi)v_d\psi_i \hat{f}\rangle}{E_{j}}$. 
  Notice 
  \begin{equation}
     \langle \tau(\bxi)v_d\psi_i\hat{f}\rangle = 
     \left\langle \tau(\bxi)v_d\psi_ih\left(
     \sum_{k=0}^{M-1}\eta_{k} \psi_k
     \right)\right\rangle, \quad 0\leq i\leq M-1,
  \end{equation}
  and the constraints provided by the moments satisfying
  $\alpha\in\mI$ are
  \begin{equation}
    E_{i} = \langle \tau(\bxi)\psi_i\hat{f}\rangle = 
     \left\langle \tau(\bxi)\psi_ih\left(
     \sum_{k=0}^{M-1}\eta_{k}\psi_k
     \right)\right\rangle, \quad 0\leq i\leq M-1,
  \end{equation}
  we have that when $d=1,2,\cdots,D$ and $0\leq i,j\leq M-1$, 
  \[   
    \pd{E_{i}}{\eta_{j}} = \pd{\left\langle \tau(\bxi)
    \psi_i
    h\left(\sum_{k=0}^{M-1}\eta_{k}\psi_k\right)
    \right\rangle}{\eta_{j}} = 
    \left\langle\tau(\bxi)\psi_ih'\left(
    \sum_{k=0}^{M-1}\eta_{k}\psi_k
    \right)\psi_j\right\rangle = \left\langle \tau(\bxi)\weight
    \psi_i\psi_j\right\rangle
    =\bD_{i,j},
  \]
  and 
  \[   
  \pd{\langle\tau(\bxi)v_d\psi_i\hat{f}\rangle}{\eta_{j}} = 
  \pd{\left\langle \tau(\bxi)
    v_d\psi_i
    h\left(\sum_{k=0}^{M-1}\eta_{k}\psi_k\right)
    \right\rangle}{\eta_{j}} = 
    \left\langle\tau(\bxi)v_d\psi_ih'\left(
    \sum_{k=0}^{M-1}\eta_{k}\psi_k
    \right)\psi_j\right\rangle 
    = \left\langle \tau(\bxi)\weight
    v_d\psi_i\psi_j \right\rangle
    =\bK_{d,i,j}.
  \]
  Therefore,
  $\pd{\eta_{j}}{E_{i}}$ is the
  $(i,j)$-th element of $\bD^{-1}$, 
  since $\bD$ is symmetric and positive
  definite. Furthermore, according to  
  \[
    \pd{
   \langle\tau(\bxi)v_d\psi_i\hat{f}\rangle
    }{E_{j}} =
    \sum_{k=0}^{M-1} \pd{
   \langle\tau(\bxi)v_d\psi_i\hat{f}\rangle
    }{\eta_{k}}\pd{\eta_{k}}{E_{j}}
    = \sum_{k=0}^{M-1}
  \bK_{d,i,k}\bD^{-1}_{k,j} = (\bK_d\bD^{-1})_{i,j}.
  \]
  Therefore, $\bM_d= \bK_d\bD^{-1}$, for any $d=1,2,\cdots,D$, and the
  system \eqref{eq:system_mn_multiD} is equivalent to
  \eqref{eq:systemmulti_d_matrix}.
  
  At last, since \eqref{eq:system_mn_multiD} can be written as 
  \[
    \left\langle \tau(\bxi)\psi_k \left(\pd{\hat{f}}{t} +
  \bv\cdot \nabla_{\bx}\hat{f} \right)
  \right\rangle =
  \left\langle \tau(\bxi)\psi_k\mS \right\rangle, \quad 0\leq k\leq M-1,
  \]
  \eqref{eq:system_mn_multiD} and \eqref{eq:systemmulti_d_matrix} are
  in conservation form. 
\end{proof}

Next, we consider cases where the number of parameters $\bm\eta$ is
less than the number of moments $M$. We calculate the result of
one-step Maxwellian iteration for setting $f^{(0)} = \weight$ instead of
$f^{(0)} = f_{eq}$.  Similar as 1-D case, we have the following theorem
\begin{theorem}
  \label{thm:maxwelliter_multiD}
  For the multi-dimensional kinetic equation \eqref{eq:kineticmultid}, if we
  use the weight function 
  \begin{equation}\label{eq:ansatz_mnpn_multiD}
    \weightf = h\left(\sum\limits_{i =0}^n \eta_i \psi_i\right),
  \end{equation} where $n\leq M-1$,
  and construct a weighted polynomial $\hat{f} =
  \sum_{i=0}^{M-1}f_i \psi_i \weightf$ to approximate the distribution
  function $f$ and derive a moment system $\mathcal{A}_1$, and denote
  the system \eqref{eq:systemmulti_d_matrix} by taking the weight function 
  \begin{equation}
    \label{eq:weight_mnpn_multiD}
    \weight = h'\left( \sum_{i =0}^n \eta_i \psi_i \right),
  \end{equation}
  as $\mathcal{A}_2$.  Then the results of one-step Maxwellian
  iteration of $\mathcal{A}_1$ and $\mathcal{A}_2$ are the same. 
\end{theorem}
\begin{proof}
  Denote $ \bDn$ as $\bDn_{i,j} = \langle \tau(\bxi)\psi_i
  \psi_j \weightf \rangle$, and $\bKn_d$ as
  $\bKn_{d,i,j} = \langle \tau(\bxi)v_d 
  \psi_i \psi_j \weightf\rangle$, $0\leq i,j\leq M-1$, 
  then $\mathcal{A}_1$ can be written
  as 
  \[
    \pd{\bm{E}}{t} + \sum\limits_{d=1}^D \pd{(\bKn_d
    \bDn^{-1}\bm{E})}{x_d} = \bm{S}.
  \]
  Furthermore, we can rewrite these two systems by a variable
  $\bw=(\bm\eta,\bm\eta^*)\in\bbR^{M}$, where there is an one-to-one
  map between $\bw$ and $\bm{E}$. Additionally, we assume that when
  $\hat{f}=\weightf$, $\bm\eta^*=0$.  Without loss of generality, we
  assume that $\bA := \pd{\bm E}{\bw} $ is invertible. 

  Then $\mathcal{A}_1$ and $\mathcal{A}_2$ can be respectively written
  as 
  \[ 
    \bA\pd{\bw}{t} + \sum\limits_{d=1}^D \pd{(\bKn_d\bDn^{-1}\bm
    E)}{\bw}\pd{\bw}{x_d} = \bm{S},\quad 
    \bA\pd{\bw}{t} + \sum\limits_{d=1}^D \bK_d \bD^{-1}\bA\pd{\bw}{x_d} = \bm{S},\quad 
  \]
  In order to prove that the one-step Maxwellian iteration of these
  two systems are equivalent, one only needs to prove that 
  \begin{equation} 
    \label{eq:thm_maxwell_target_multiD}
    (\bK_d \bD^{-1}\bA)_{m,l}\big|_{\bm\eta^*=0} = 
    \pd{(\bKn_d\bDn^{-1}\bm{E})_m}{\eta_l}\Big|_{\bm\eta^*=0},
    \quad
  0 \leq m \leq M-1, \quad 0\leq l\leq n, \quad 1 \leq d \leq D.
\end{equation}
  The right hand side can be calculated as 
  \begin{equation}
    \label{eq:linear_cal_multiD}
    \pd{(\sum_{i,j=0}^{M-1}\bKn_{d,m,i}\bDn^{-1}_{i,j}E_j)}{\eta_l}\Big|_{\bm\eta^*=0}
    = 
    \pd{[(\sum_{i,j=0}^{M-1}\bKn_{d,m,i}\bDn^{-1}_{i,j}E_j)|_{\bm\eta^*=0}]}{\eta_l} 
    =
    \pd{\bKn_{d,m,0}}{\eta_l}\Big|_{\bm\eta^*=0} =
    \bK_{d,m,l}|_{\bm\eta^*=0},
\end{equation}
where $\psi_0=1$, thus 
  \[
  (\sum_{j=0}^{M-1}\bDn^{-1}_{i,j}E_j)|_{\bm\eta^*=0} =
  \begin{cases}
    0,& i\neq 0,\\
    1, & i =0.
  \end{cases}
  \]
  and 
  \[  
     \pd{\bKn_{d,i,0}}{\eta_l} = \bK_{d,i,l}
     \quad 
     \pd{\bDn_{i,0}}{\eta_l} = \bD_{i,l},\quad 0\leq i\leq M-1, 0\leq
     l\leq n
  \]
  are applied. On the other hand, notice that 
  \[
  \bA_{i,l}|_{\bm\eta^*=0} = \pd{E_i}{\eta_l}\Big|_{\bm\eta^*=0} = 
  \pd{\bDn_{i,0}}{\eta_l}\Big|_{\bm\eta^*=0} =
  \bD_{i,l}|_{\bm\eta^*=0}, \quad 0\leq i\leq M-1,0\leq l\leq n.
  \]
  Therefore, 
  \begin{equation} 
    (\bK_d\bD^{-1}\bA)_{m,l}\Big|_{\bm\eta^*=0} = \bK_{d,m,l}, 
  \quad 0\leq m\leq M-1,
  0\leq l\leq n,
  1\leq d\leq D,
\end{equation}
  Thus \eqref{eq:thm_maxwell_target_multiD} holds.
\end{proof}


%% file: examplemultid.tex

\section{Applications in multi-dimensional case}\label{sec:examplemultid}
In this section, we will list come existing hyperbolic moment models
in multi-dimensional case. They can be regarded as special cases of our new
framework. Furthermore, we will also propose some new hyperbolic
models. 

\subsection{Boltzmann equation}\label{sec:Boltzmannmd}
In \cref{sec:Boltzmann1d}, we consider Boltzmann equation in 1-D case
and verify the equality of our framework and some existing hyperbolic
moment models, such as one-dimensional HME and maximum entropy model. 
In this section, we will prove that in multi-dimensional case, the HME
and maximum entropy model can still be included in our new framework
\eqref{eq:systemmulti_d_matrix}. 

First, the multi-dimensional Boltzmann equation reads 

\begin{equation}
  \pd{f}{t} + \bxi \cdot \nabla_{\bx}f = \mS(f), 
\end{equation}
where $f=f(t,\bx,\bxi)$ is the distribution function, with
$\bx\in\bbR^D$, and $\bxi\in\bbR^D$. Compared with
\eqref{eq:kineticmultid}, we have $\bv(\bxi)=\bxi$.

The collision term $\mS(f)$ depicts the
interactions between particles, whose form is given by
\cite{Grad, Struchtrup} 
\[
\mS(f)=\int_{\bbR^D}\int_{S_+^{D-1}}(f'f_*'-ff_*)B(|\bxi-\bxi_*|,\sigma)\dd\bm{n}\dd\bxi_*,
\]
where $f,f',f_*,f'_*$ are the shorthand notations for $f(t,\bx,\bxi)$,
$f(t,\bx,\bxi')$, $f(t,\bx,\bxi_*)$ and $f(t,\bx,\bxi'_*)$,
respectively. $(\bxi,\bxi_*)$ and $(\bxi',\bxi'_*)$ are the velocities
before and after collision, $\bm{n}$ represents the collision angle,
and $B(|\bxi-\bxi_*|,\sigma)$ is the collision kernel.

The local equilibrium is given as Maxwellian \cite{Maxwell}, formulated as 
\[ 
f_{eq} = \dfrac{\rho}{\sqrt{2\pi\theta}^D}\exp\left(
-\dfrac{|\bxi-\bm u|^2}{2\theta} \right),
\]
where 
\begin{equation}
  \rho = \langle f\rangle, \quad \rho \bm{u} = \langle \bxi
  f\rangle,\quad 
  \dfrac{1}{2}\rho|\bm u|^2 + \dfrac{D}{2}\rho\theta = \left\langle
  \dfrac{|\bxi|^2}{2}f \right\rangle.
\end{equation}

\subsubsection{HME model} \label{sec:HME}
The Grad's moment model \cite{Grad} suggests to use a weighted polynomial to approximate the
distribution function, where the weight function is taken as the local
equilibrium,
\[
\tilde{\omega}^{[\bm\eta]} =  
  \dfrac{\rho}{\sqrt{2 \pi \theta}^D}
  \exp\left(-\frac{|\bxi-\bm u|^2}{2\theta}\right),
\] 
By taking $\bm\eta=\left(\ln\left(\dfrac{\rho}{\sqrt{2\pi\theta}^D}\right)
-\dfrac{|\bm u|^2}{2\theta}, \dfrac{\bm
u}{\theta},-\dfrac{1}{2\theta}\right)^T$, one can rewrite the weight
function into the form \eqref{eq:ansatz_mnpn_multiD}.

The moments considered in $N$-th order HME model is $E_k$, $0\leq k\leq
M-1$, where $M = \binom{N+D}{D}$, and 
\[ 
  E_{\mN(\alpha)} = \langle\bxi^{\alpha}f\rangle,\quad |\alpha|\leq N.
\]
Compared with \eqref{eq:moment_multid}, we have $\tau(\bxi)=1$.
\begin{remark}
  In this example, we require that 
  \[  
    \mN\left(
    \left\{ \alpha\in\bbN^D : |\alpha|\leq N\right\}
    \right) = \{0,1,2,\cdots,M-1\}.
  \]
  For example, we can let \cite{Fan_new} 
  \[ 
  \mN(\alpha) = \sum_{i=1}^{D}
  \binom{\sum_{k=D-i+1}^D\alpha_k+i-1}{i}.
  \]
  In the following discussion of this paper, we will use
  $(\cdot)_{\alpha}$ to represent the $\mN(\alpha)$ of vector
  $(\cdot)$, and $(\cdot)_{\alpha,\beta}$ to represent the
  $(\mN(\alpha),\mN(\beta))$-th element of matrix $(\cdot)$.
\end{remark}
The ansatz for $f$ is 
\[ 
\hat{f} = 
\sum_{|\alpha|\leq N}
f_{\alpha}\mHe_{\alpha}
 = \tilde{\omega}^{[\bm\eta]} \sum_{|\alpha|\leq N}
f_{\alpha}\He_{\alpha},
\]
where $\mHe_{\alpha}(\bxi)=\He_{\alpha}(\bxi)\weightf(\bxi)$ is orthogonal basis, and
$\He_{\alpha}(\bxi)$ is generalized Hermite polynomial, which is orthogonal
polynomials with respect to the weight function $\weightf$. Actually,
according to \cite{NRxx,Fan_new}, we have 
\[  
   \He_{\alpha}(\bxi) = \prod_{d=1}^{D}\He_{\alpha_d}(\xi_d)
   =\prod_{d=1}^{D}\text{He}_{\alpha_d}\left(
   \dfrac{\xi_d-u_d}{\sqrt{\theta}} \right),
\]
where $\He_{k}$ and $\text{He}_k$ are defined in \cref{sec:HME1d}.

With a
little abuse of the notation, we assume $\He_{\alpha}$ are unit
orthogonal polynomials, and $\mHe_{\alpha}$ is unit orthogonal basis.

\begin{property}\label{pro:3dhmeequal}
  Multi-dimensional HME model is equivalent to the moment model
  \eqref{eq:systemmulti_d_matrix} 
by taking the weight function as 
\[
\weight =  
  \dfrac{\rho}{\sqrt{2 \pi \theta}^D}
  \exp\left(-\frac{|\bxi-\bm u|^2}{2\theta}\right).
  \]
\end{property}
\begin{proof}
According to the statements in \cite{framework}, multi-dimensional HME
system can be written as 
\begin{equation}
  \label{eq:HMEmd}
  \tilde{\bD}\pd{\bw}{t} +
  \sum_{d=1}^{D}\tilde{\bM}_d\tilde{\bD}^{-1}\pd{\bw}{x_d} =
  \tilde{\bm{S}},
\end{equation}
where $\tilde{\bm{S}}\in\bbR^{M}$ is a vector whose $\mN(\alpha)$-th element is 
\[
\tilde{\bm{S}}_{\alpha} = \langle \He_{\alpha}\mathcal{S}\rangle = \int_{\bbR^D}
  \He_{\alpha}\mathcal{S}\dd\bxi.
\]
Moreover,
$\bw\in\bbR^{M}$ is a vector corresponding to the ansatz of the
distribution function $\hat{f}$. Furthermore, we have \cite{framework}
for $|\alpha|,|\beta|\leq N$ and $d=1,2,\cdots,D$, 
\begin{equation}
  \label{eq:matrixD_oldframework}
  \tilde{\bD}_{\alpha,\beta} = 
  \int_{\bbR^D} \He_{\alpha}\pd{\hat{f}}{\bw_{\beta}} \dd\bxi,
\end{equation}
and 
\begin{equation}
  \label{eq:matrixM_oldframework}
  \tilde{\bM}_{d,\alpha,\beta} = 
  \int_{\bbR^D} \xi_d\He_{\alpha}\He_{\beta} \weightf \dd\bxi.
\end{equation}

On the other hand, in our new framework, take the weight function as 
\[
\weight  = 
  \dfrac{\rho}{\sqrt{2 \pi \theta}^D}
  \exp\left(-\frac{|\bxi-\bm u|^2}{2\theta}\right),
\]
and the basis function $\psi_{\mN(\alpha)} = \bxi^{\alpha}$,
$|\alpha|\leq N$, then 
the resulting system is \eqref{eq:systemmulti_d_matrix}. Now we prove
  that these two system are the same.  First, notice that
  $\He_{\alpha}, |\alpha|\leq N$ and $\psi_{\mN(\alpha)} =
  \bxi^{\alpha},|\alpha|\leq N$ are two bases of 
\[ 
  \{\bxi^{\alpha}:|\alpha|\leq N\},
\]
there exists a matrix $\bT\in\bbR^{M\times M}$, which satisfies 
\[   
   (\He_{0}, \He_{e_1}, \dots, \He_{\alpha^*})^T 
    = \bT (\bxi^{0},\bxi^{e_1},\dots,\bxi^{\alpha^*})^T,
\]
where $\alpha^*$ is the last element in
$\mI=\{\alpha\in\bbN^D:|\alpha|\leq N\}$, 
i.e. $\mN(\alpha^\ast)
= M-1$.

According to the definition of $\tilde{\bm{S}}$ and $\bm{S}$, we have 
\begin{equation}
  \label{eq:equal_S}
  \tilde{\bm{S}} = \bT \bm{S}.
\end{equation}
Moreover, let matrix $\tilde{\bD}'$ be defined as 
\begin{equation}
  \label{eq:matrixD'_oldframework}
  \tilde{\bD}'_{\alpha,\beta} = 
  \int_{\bbR^D} \bxi^{\alpha}\pd{\hat{f}}{\bw_{\beta}} \dd\bxi,
\end{equation}
then we have $\tilde{\bD}'\pd{\bw}{s} = \pd{\bm{E}}{s}$,
$s=t,x_1,x_2,\cdots,x_D$, and $\tilde{\bD} = \bT\tilde{\bD}'$.
Moreover,  
\[  
\tilde{\bM}_d =  \bT \bK_d\bT^{T},
\]
and 
\[
\bD = \bT^{-1}\bT^{-T},
\]
where $\bD$ and $\bK_d$ are as defined in \eqref{eq:systemmulti_d_matrix}.
Therefore, we have 
\[
\tilde{\bM}_d = \bT \bK_d \bD^{-1}\bT^{-1}, 
\]
and \eqref{eq:HMEmd} can be rewritten as 
\[  
\bT\pd{\bm E}{t} + \sum_{d=1}^{D}\bT\bK_d\bD^{-1}\pd{\bm E}{x_d} =
   \bT\bm{S},
\]
which is equivalent to \eqref{eq:systemmulti_d_matrix}. Therefore,
multi-dimensional HME can be included by our framework.
\end{proof}
\subsubsection{HME model for 13-moment system}\label{sec:13m}
Grad's 13-moment \cite{Grad} is the most famous model in the Grad-type
moment hierarchy, and the HME model
for 13-moment system is its hyperbolic regularized version
\cite{Fan2015}. Different from the Grad's moment method mentioned in
the previous section, the
moments considered here do not correspond to $\langle
\tau(\bxi)\bxi^{\alpha} f\rangle$, with $|\alpha|\leq N$. In this
section, we will introduce the HME model for 13-moment system and put
it into our new framework \eqref{eq:systemmulti_d_matrix}. 

In this case, the dimension $D=3$, and the moments are defined as 
\[  
   \bm{E} = (\langle f \rangle, \langle \xi_1 f\rangle,
   \langle \xi_2 f\rangle, \langle \xi_3 f\rangle, 
   \langle\xi_1^2 f\rangle, \langle \xi_1\xi_2 f\rangle, \langle
   \xi_1\xi_3 f\rangle,
   \langle \xi_2^2 f\rangle, \langle\xi_2\xi_3 f\rangle,
   \langle \xi_3^2f\rangle, 
   \langle \Vert\bxi\Vert^2\xi_1f\rangle,
   \langle \Vert\bxi\Vert^2\xi_2f\rangle,
   \langle \Vert\bxi\Vert^2\xi_3f\rangle
   )^T,
\]
which implies that $\tau(\bxi)=1$, and  
\[  
\bm{\psi} =
(1,\xi_1,\xi_2,\xi_3,\xi_1^2,\xi_1\xi_2,\xi_1\xi_3,\xi_2^2,\xi_2\xi_3,\xi_3^2,
\Vert\bxi\Vert^2\xi_1,
\Vert\bxi\Vert^2\xi_2,
\Vert\bxi\Vert^2\xi_3)^T.
\]

$\pd{f}{x_d}$ is approximated by the ansatz $\hat{g}_d$, written as  
\[  
\hat{g}_d = \weight \sum_{i=0}^{M-1} g_i\psl_i.
\]
then the moment system can be written as 
\begin{equation}
  \label{eq:13m}
\pd{\bm{E}}{t} + \sum_{d=1}^{3}\bK_d\bD^{-1}\pd{\bm{E}}{x_d} = \bm{S},
\end{equation}
where the $(i,j)$-th element of $\bD$ and $\bK_d$ are 
\[  
\bD_{i,j} = \langle \psl_i\psl_j\weight\rangle,\quad 
\bK_{d,i,j} = \langle \xi_d\psl_i\psl_j\weight\rangle.
\]
Using the same technique as in the proof of \cref{pro:3dhmeequal}, it
is not difficult to see that \eqref{eq:13m} is equivalent to the HME13
system proposed in \cite{Fan2015}.

\subsubsection{Maximum entropy model}
The moments considered in the full moment maximum entropy model
\cite{levermore1996moment} 
for multi-dimensional case 
are described by $\psi_{\mN(\alpha)} = \bxi^{\alpha}$, $|\alpha|\leq
N$.

The maximum entropy model uses the following ansatz 
for the distribution function:
\begin{equation}\label{eq:mep_ansatz_multiD}
  \hat{f}(t,\bx,\bxi) = \exp\left(\sum\limits_{k=0}^{M-1}
  \eta_{k}(t,\bx) \psi_k
  \right),
\end{equation}
which corresponds to $h(\zeta) = \exp(\zeta)$ in \cref{thm:special_case_multiD}.
Therefore, according to \cref{thm:special_case_multiD}, take the weight
function as 
\[  
\weight = h'\left( \sum\limits_{k=0}^{M-1} \eta_{k}\psi_k\right) = 
\exp\left( \sum_{k=0}^{M-1}\eta_{k}\psi_k \right),
\]
and we know the two moment models are the same.

\subsection{Three-dimensional radiative transfer equation}
Three-dimensional radiative transfer equation is written as 
\begin{equation}
  \label{eq:3drte}
  \dfrac{1}{c}\pd{I}{t} + \bxi\cdot\nabla_{\bx}I = \mathcal{S},
\end{equation}
where the distribution function (specific intensity)
$I=I(t,\bx,\bxi)$, and $t\in\bbR^+$, $\bx\in\bbR^3$,
$\bxi\in\bbS^2$. The right hand side $\mathcal{S}$ depicts the
interactions with other photons and  the background medium. In RTE, we
have $\tau(\bxi)=1$ and $\bv(\bxi)=\bxi$.

In radiative transfer equation, notice that $\Vert \bxi\Vert = 1$, we
have that 
\[   
  \sum_{d=1}^{3} \langle \bxi^{\alpha}\xi_d^{2} f\rangle = 
  \langle \bxi^{\alpha} f\rangle,\quad \alpha\in\bbN^3.
\]
Thus the moments $\langle \bxi^{\alpha}f\rangle$, $|\alpha|\leq N$ 
can be expressed by 
$ 
  \langle \bxi^{\alpha}f\rangle$,$\alpha\in\mI$, with 
  \[
   \mI = \{\alpha\in\bbN^3:|\alpha|\leq N,\alpha_3\leq 1 \},
  \]
  and $M=\#\mI = (N+1)^2$. A map from $\mI$ to $\{0,1,2,\cdots,M-1\}$
  can be  defined as 
  \[
    \mN(\alpha) = (\alpha_1+\alpha_2+\alpha_3)^2 +
    (\alpha_1+\alpha_2+\alpha_3+1)\alpha_3  + \alpha_2,
  \]
  and the basis function is $\psi_{\mN(\alpha)} = \bxi^{\alpha}$,
  $\alpha\in\mI$. 

\subsubsection{\MN model}
To derive the \MN model for gray RTE, one uses the following ansatz for the
distribution function:
\begin{equation}\label{eq:mn_ansatz_multiD}
  \hat{f}(t,\bx,\bxi) = \left(\sum\limits_{k=0}^{M-1}
  \eta_{k}(t,\bx) \psi_k
  \right)^{-4},
\end{equation}
which corresponds to $h(\zeta) = \dfrac{1}{\zeta^4}$ in
\cref{thm:special_case_multiD}.

On the other hand, for the \MN model for monochromatic radiative transfer, 
the ansatz to approximate the distribution function is given by 
\begin{equation}
  \hat{f} (t,\bx,\bxi) = \left[\exp\left(\sum\limits_{k=0}^{M-1}
  \eta_{k}(t,\bx) \psi_k \right)-
  1\right]^{-1},
\end{equation}
which corresponds to $h(\zeta) = \dfrac{1}{\exp(\zeta)-1}$ in
\cref{thm:special_case_multiD}.

Therefore, according to \cref{thm:special_case_multiD}, we take the weight function 
for the grey RTE as 
\[ 
\weight = \dfrac{1}{\left(\sum\limits_{k=0}^{M-1}
  \eta_{k}\psi_k \right)^5},
\]
and the weight function for the monochromatic RTE is 
\[
\weight = \dfrac{\exp\left(\sum\limits_{k=0}^{M-1}
\eta_{k}\psi_k\right)}{
\left(\exp\left(\sum\limits_{k=0}^{M-1}\eta_{k}
  \psi_k\right)-1\right)^2}.
\]
Then the \MN model for the grey RTE and the monochromatic RTE can also
regarded as examples of our new framework.

\subsubsection{\HMPN model}
In \cite{3DMPN}, the \HMPN model was extended to 3D case, and a
globally hyperbolic moment model for radiative transfer equation was
carried out.

By taking 
\[  
\weight = \dfrac{1}{(\eta_0+\eta_1\xi_1+\eta_2\xi_2+\eta_3\xi_3)^5},
\]
using the same technique in the proof of \cref{pro:3dhmeequal}, we
know \eqref{eq:systemmulti_d_matrix} is equivalent to 3D \HMPN model
in \cite{3DMPN}. Furthermore, according to \cref{thm:conservation_md},
the first $(N-1)$-order moments can be written as conservation law,
and according to 
\cref{thm:maxwelliter_multiD}, we know the one-step Maxwellian
iteration of the 3D \HMPN model preserves the result of the one-step
Maxwellian iteration of the 3D \MPN model, which will be changed by a
direct application of the hyperbolic regularization in
\cite{Fan2015,framework}.

\subsection{13-moment model for quantum gas}

The quantum Boltzmann equation, as well as the well-known
Uehling-Uhlenbeck equation \cite{uehling1933}, is written as 
\[ 
     \pd{f}{t} + \bxi\cdot \nabla_{\bx} f = \mathcal{S}(f),
\]
where the collision term is defined as 
\[
\mathcal{S}(f) = \int_{\bbR^3}\int_{0}^{2\pi}\int_{0}^{\pi}
[(1-\theta f)(1-\theta f_*)f'f'_*-(1-\theta f')(1-\theta f'_*)ff_*]
g\sigma\sin \chi\dd\chi\dd\varepsilon\dd\bxi_*.
\]
Here $f,f',f_*,f'_*$ are the shorthand notations for $f(t,\bx,\bxi)$,
$f(t,\bx,\bxi')$, $f(t,\bx,\bxi_*)$ and $f(t,\bx,\bxi'_*)$,
respectively. $(\bxi,\bxi_*)$ and $(\bxi',\bxi'_*)$ are the velocities
before and after collision. $\varepsilon$ is the scattering angle,
$\chi$ is the deflection angle, $g=|\bxi-\bxi_*|$, and $\sigma$ is the
differential cross section. $\theta=1,0,-1$ corresponds to Fermion,
classical gas and Boson, respectively. Since classical gas has been
studied in the Boltzmann equation in \cref{sec:Boltzmannmd}, in this example, 
we focus on Fermion and Boson. Compared with \eqref{eq:kineticmultid},
we have $\bv(\bxi)=\bxi$.

The macroscopic variables can be defined as 
\[
\rho = \dfrac{m}{\hat{h}^3}\int_{\bbR^3} f\dd\bxi, \quad 
\rho \bm u = \dfrac{m}{\hat{h}^3}\int_{\bbR^3} \bxi f\dd\bxi, 
 \quad  
 p  = \dfrac{m}{3\hat{h}^3}\int_{\bbR^3} |\bxi-\bm u|^2 f\dd\bxi, 
\]
where $m$ is the mass of the particle, $\hat{h}=h/m$, and $h$ is
Planck's constant. 

The thermodynamic equilibrium is 
\begin{equation}\label{eq:quantum_feq}
f_{eq} = \dfrac{1}{\zz^{-1}\exp\left( \frac{|\bxi-\bm u|^2}{2RT}
\right)+ \theta},
\end{equation}
where $\zz$ and $RT$ is related to $\rho$ and $p$ as  
\begin{equation}\label{eq:relation_rhopzt}
    \rho = \constantsymbol\sqrt{2\pi\Temperature}^3\Lihalf{3},\quad 
    p=\constantsymbol\sqrt{2\pi\Temperature}^3\Temperature \Lihalf{5},
\end{equation}
and $\Li_s:=-\theta\Li_s(-\theta\zz)$ is the polylogarithm.
For the special case $\theta=0$, let $\Li_s=\zz$.

The considered moments in the quantum Grad's 13-moment system \cite{QuantumGrad13} 
are 
\[  
   \bm{E} = (\langle f \rangle, \langle \xi_1 f\rangle,
   \langle \xi_2 f\rangle, \langle \xi_3 f\rangle, 
   \langle\xi_1^2 f\rangle, \langle \xi_1\xi_2 f\rangle, \langle
   \xi_1\xi_3 f\rangle,
   \langle \xi_2^2 f\rangle, \langle\xi_2\xi_3 f\rangle,
   \langle \xi_3^2f\rangle, 
   \langle \Vert\bxi\Vert^2\xi_1f\rangle,
   \langle \Vert\bxi\Vert^2\xi_2f\rangle,
   \langle \Vert\bxi\Vert^2\xi_3f\rangle
   )^T,
\]
which implies that $\tau(\bxi)=1$, and  
\[  
\bm{\psi} =
(1,\xi_1,\xi_2,\xi_3,\xi_1^2,\xi_1\xi_2,\xi_1\xi_3,\xi_2^2,\xi_2\xi_3,\xi_3^2,
\Vert\bxi\Vert^2\xi_1,
\Vert\bxi\Vert^2\xi_2,
\Vert\bxi\Vert^2\xi_3)^T.
\]

On the other hand, in \cite{QuantumGrad13}, the authors proposed an ansatz for the distribution 
\[
  f_{G13} = f_{eq} \sum_{i=0}^{12} f_i\psi_i,
\]
where $f_i$ are coefficients.  It was pointed out in
\cite{di2016quantum} that quantum Grad's 13-moment system is not
hyperbolic and a direct application of the classical framework in
\cite{framework, Fan2015} will change the NSF law of the quantum
Grad's 13-moment system. Therefore, researchers in
\cite{di2016quantum} proposed a new regularized 13-moment system, by
splitting the expansion into the equilibrium and non-equilibrium part
and then applying the framework in \cite{framework, Fan2015}. However, this
method can not be extended to higher order moment model, in which case
the conservation of most moments can not be preserved. In this paper,
based on the method in \cref{sec:modelmultid}, we propose a new
hyperbolic regularization of the quantum Grad's 13-moment system.  

The weight function is taken as 
\[
\weight = \dfrac{\zz^{-1}\exp\left( \frac{|\bxi-\bm u|^2}{2RT} \right)}
{\left(\zz^{-1}\exp\left( \frac{|\bxi-\bm u|^2}{2RT}
\right)+ \theta\right)^2},
\]
a globally hyperbolic moment system is obtained.  According to
\cref{thm:conservation_md}, we have the conservation of $\rho$, $\rho
\bm{u}$, and $p$.  For higher order moment models, we can get that
higher order moments satisfy equations in conservation form, which can
not be ensured by the moment system proposed in \cite{di2016quantum}.
On the other hand, according to \cref{thm:maxwelliter_multiD}, the NSF
law of the new moment system is the same as the quantum Grad's
13-moment system, since $f_{eq}$ and $\weight$ satisfy the conditions
in \cref{thm:maxwelliter_multiD}.


%% file: conclusion.tex
\section{Conclusion}\label{sec:conclusion}

We proposed an improved framework for globally hyperbolic model
reduction of kinetic equations by a new closure approach. Different
from previous works, most of which derive moment closure by specifying
an ansatz for the distribution function, our method directly gives the
closing relationship of the flux gradient by approximating the spatial
derivative of the distribution function with weighted
polynomials. Besides being globally hyperbolic, moment models derived
in this way satisfy nice physical properties, such as except for the
highest order moments, all the other moments satisfy equations in
conservation form. Also, the results of the Maxwellian iteration after
the first step remain unchanged, ensuring the models satisfy important
properties like the Navier-Stokes-Fourier law. As shown in previous
sections, most of the globally hyperbolic moment models developed in
the previous literature could be regarded as special cases of this new
framework, and we also derived new moment models, such as for the
monochromatic radiative transfer equation, phonon Boltzmann equation
and quantum gases, which are either the first moment models for such
equations with global hyperbolicity, or improvements on previously
developed ones. This framework provides us with a new tool to derive
in a routine way globally hyperbolic moment models which preserve nice
physical properties. Numerical schemes specially designed for this
type of moment model will be further investigated.

\section*{Acknowledgements}
The work of R.L. and L.Z. is partially supported by Science Challenge
Project, No. TZ2016002 and the National Natural Science Foundation in
China (Grant No.  11971041). The work of W.L. is partially supported
by Science Challenge Project (NO. TZ2016002) and National Natural
Science Foundation of China (12001051).
